\documentclass[conference]{IEEEtran}

\usepackage{threeparttable}
\usepackage[T1]{fontenc}    
\usepackage{hyperref}       
\usepackage{url}            
\usepackage{booktabs}       
\usepackage{amsfonts}       
\usepackage{nicefrac}       
\usepackage{microtype}      
\usepackage{times}
\usepackage{helvet}
\usepackage{courier}
\usepackage{amsmath,amsfonts,amssymb,bbm,amsthm} 
\usepackage{multirow}
\usepackage{extarrows}
\usepackage{xcolor}
\usepackage{algorithm}
\usepackage{algorithmic}
\usepackage{graphicx}
\usepackage{color}
\usepackage{caption}
\usepackage{setspace}

\newtheorem{theorem}{Theorem}
\newtheorem{lemma}{Lemma}

\newtheorem{corollary}{Corollary}
\newtheorem{definition}{Definition}
\newtheorem{remark}{Remark}

\allowdisplaybreaks

\newcommand\bovermat[2]{%
	\makebox[0pt][l]{$\smash{\overbrace{\phantom{%
					\begin{matrix}#2\end{matrix}}}^{\text{#1}}}$}#2}

\onecolumn
\begin{document}

\title{Fundamental Limits of Coded Linear Transform}

\author{\IEEEauthorblockN{Sinong Wang$^{*}$, Jiashang Liu$^{*}$, Ness Shroff$^{*\dag}$ and Pengyu Yang$^{\intercal}$}
	\IEEEauthorblockA{$^{*}$Department of Electrical and Computer Engineering, The Ohio State University\\
		Email: 	\{wang.7691, liu.3992\}@osu.edu\\
		$^{\dag}$Department of Computer Science and Engineering, The Ohio State University\\
		Email: shroff.11@osu.edu\\
		$^{\intercal}$Department of Mathematics, The Ohio State University\\
		Email: yang.2214@osu.edu}}

\maketitle

\begin{abstract}
In large scale distributed linear transform problems, coded computation plays an important role to effectively deal with ``stragglers'' (distributed computations that may get delayed due to few slow or faulty processors).  We propose a coded computation strategy, referred to as \emph{diagonal code}, that achieves the optimum recovery threshold and the optimum computation load. This is the first code that simultaneously achieves two-fold optimality in coded distributed linear transforms. Furthermore, by leveraging the idea from random proposal graph theory, we design two \emph{random codes} that can guarantee optimum recovery threshold with high probability but with much less computation load. These codes provide order-wise improvement over the state-of-the-art. Moreover, the experimental results show significant improvement compared to both uncoded and existing coding schemes.	
\end{abstract}

\section{Introduction}

In this paper, we propose a novel coding theoretical framework for speeding up the computation of linear transforms $\mathbf{y}=\mathbf{Ax}$ across multiple processing units, where  $\mathbf{A}\in\mathbb{R}^{r\times t}$ and $\mathbf{x}\in\mathbb{R}^t$. This problem is the key building block in machine learning and signal processing problems, and has been used in a large variety of application areas. Optimization-based training algorithms such as gradient descent in regression and classification problems, backpropagation algorithms in deep neural networks, require the computation of large linear transforms of high-dimensional data. It is also the critical step in the dimensionality reduction techniques such as Principal Component Analysis (PCA) and Linear Discriminant Analysis (LDA). Many such applications have large-scale datasets and massive computational tasks that force practitioners to adopt distributed computing frameworks such as Hadoop~\cite{dean2008mapreduce} and Spark~\cite{zaharia2010spark} to increase the learning speed. 

Classical approaches of distributed linear transforms rely on dividing the input matrix $\mathbf{A}$ equally among all available worker nodes, and the master node has to collect the results from all workers to output $\mathbf{y}$. As a result, a major performance bottleneck is the latency incurred in waiting for a few slow or faulty processors -- called ``stragglers'' to finish their tasks~\cite{dean2013tail}.  Recently, forward error correction and other coding techniques have shown to provide an effective way to deal with the ``straggler'' in the distributed computation tasks~\cite{dutta2016short,lee2017speeding,lee2017coded,tandon2017gradient,yu2017polynomial,li2018fundamental,wang2018coded}. By exploiting coding redundancy, the vector $\mathbf{y}$ is recoverable even if all workers have not finished their computation, thus reducing the delay due to straggle nodes. One main result in this paper is the development of optimal codes to deal with stragglers in computing distributed linear transforms, which provides an optimum recovery threshold and optimum computation redundancy.

In a general setting with $m$ workers, the input matrix $\mathbf{A}$ is divided into $n$ submatrices $\mathbf{A}_i\in\mathbb{R}^{\frac{r}{n}\times t}$ along the row side. Each worker stores $1/n$ fraction of $\mathbf{A}$, denoted by $\tilde{\mathbf{A}}_i$, which is an arbitrary function of $[\mathbf{A}_i]_{i=1}^{n}$. Each worker computes a partial linear transform and returns the result to the master node. By carefully designing the coded computation strategy, i.e., designing $\tilde{\mathbf{A}}_i$, the master node only needs to wait for the fastest subset of workers before recovering $\mathbf{y}$. Given a computation strategy, the \emph{recovery threshold} is defined as the minimum number of workers that the master needs to wait for in order to compute $\mathbf{y}$, and the \emph{computation load} is defined as the total number of submatrices each worker access. Based on the above notation, our main focus is

\emph{What is the minimum recovery threshold and computation load for coded distributed linear transforms? Can we find an optimum computation strategy that achieves both optimality?}

To that end, there have been works that have investigated the problems of reducing the recovery threshold and computational load for coded distributed linear transforms by leveraging ideas from coding theory. The first work~\cite{lee2017speeding} proposes a maximum distance separable (MDS) code-based linear transform. We illustrate this approach using the following example. Consider a distributed system with $3$ workers, the coding scheme first splits the matrix $\mathbf{A}$ into two submatrices, i.e., $\mathbf{A}=[\mathbf{A}_1;\mathbf{A}_2]$. Then each worker computes $\mathbf{A}_1\mathbf{x}$, $\mathbf{A}_2\mathbf{x}$ and $(\mathbf{A}_1+\mathbf{A}_2)\mathbf{x}$.  The master node can compute $\mathbf{A}\mathbf{x}$ as soon as \textbf{any} $2$ out of the $3$ workers finish, and therefore achieves a recovery threshold of $2$ and a computation load of $4$. More specifically, it is shown that MDS code achieves a recovery threshold of $\Theta(m)$ and computation load $O(mn)$. An improved scheme proposed in~\cite{dutta2016short}, referred to as the \emph{short MDS code}, can offer a larger recovery threshold but with a constant reduction of computation load by imposing some sparsity of the encoded submatrices.  More recently, the work in~\cite{yu2017polynomial} designs a type of \emph{polynomial code}. It achieves the recovery threshold of $n$ that is independent of number of workers. However, it incurs a computation load of $mn$. In other words, each worker of polynomial code has to access the information of entire data matrix $\mathbf{A}$.

In this paper, we show that, given number of partitions $n$ and number of stragglers $s$, the minimum recovery threshold is $n$ and the minimum computation load is $n(s+1)$. We design a novel coded computation strategy, referred to as the \emph{$s$-diagonal code}. It achieves both optimum recovery threshold and optimum computation load, which significantly improves the state-of-the-art. We also design a hybrid decoding algorithm between the peeling decoding and Gaussian elimination. We show that the decoding time is nearly linear time of the output dimension $O(r)$. 

Moreover, we define a new \emph{probabilistic recovery threshold} performance metric that allows the coded computation scheme to provide the decodability in a large percentage (with high probability) of straggling scenarios. Based on this new metric, we propose two kinds of random codes, $p$\emph{-Bernoulli code} and $(d_1,d_2)$\emph{-cross code}, which can probabilistically  guarantee the same recovery threshold but with much less computation load. 

In the $p$\emph{-Bernoulli code}, each worker stores a weighted linear combination of several submatrices such that each $\mathbf{A}_i$ is chosen with probability $p$ independently. We show that, if $p=2\log(n)/n$, the $p$\emph{-Bernoulli code} achieves the probabilistic recovery threshold of $n$ and a constant communication and computation load $O(n\log(n))$ in terms of the number of stragglers. In the $(d_1,d_2)$\emph{-cross code}, each worker first randomly and uniformly chooses $d_1$ submatrices, and each submatrix is randomly and uniformly assigned to $d_2$ workers; then each worker stores a linear combination of the assigned submatrices and computes the corresponding linear transform. We show that, when the number of stragglers $s=\text{poly}(\log(n))$, the $(2,3)$-cross code, or when $s=O(n^{\alpha}),\alpha<1$, the $(2,2/(1-\alpha))$-cross code achieves the probabilistic recovery threshold of $n$ and a much lower communication and computation load $O(n)$. We compare our proposed codes and the existing codes in Fig.~\ref{fig:summary}.

We finally implement and benchmark the constructed code in this paper on the Ohio Super Computing Center~\cite{OhioSupercomputerCenter1987}, and empirically demonstrate its performance gain compared with existing strategies.

\begin{figure}[t]
	\vskip -0.05in
	\begin{center}
		\centerline{\includegraphics[width=4.5in]{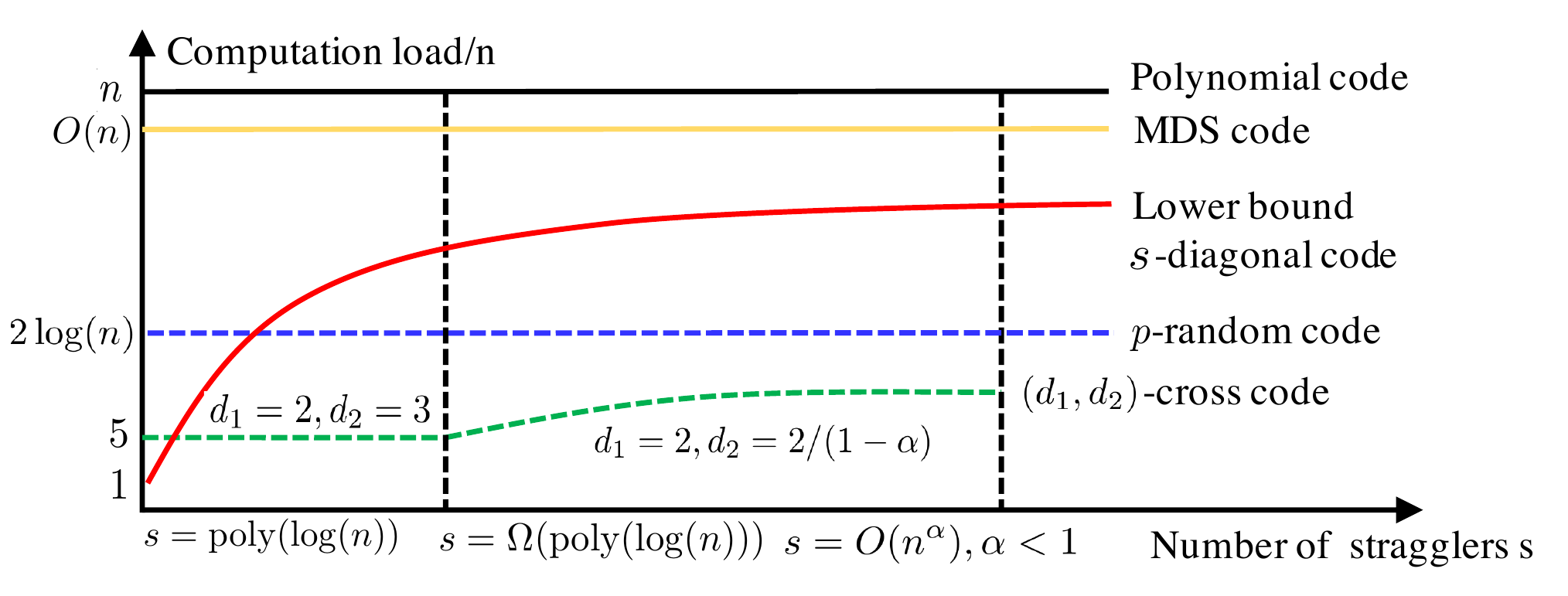}}
		\vskip -0.0in
		\caption{Comparison of computation load among the existing coded schemes and constructed codes in this paper. All the deterministic codes except MDS code achieve the optimum recovery threshold of $n$. The random codes achieves the recovery threshold of $n$ with high probability.}
		\vskip -0.4in
		\label{fig:summary}
	\end{center}
\end{figure}

\section{Problem Formulation}

We are interested in distributedly computing a linear transform with matrix $\mathbf{A}\in \mathbb{R}^{r\times t}$ and input vector $\mathbf{x}\in\mathbb{R}^{t}$ for some integers $r,t$. The matrix $A$ is evenly divided along the row side into $n$ submatrices.
\begin{equation}
\mathbf{A}^T=[\mathbf{A}_1^T, \mathbf{A}_2^T, \mathbf{A}_3^T,\ldots, \mathbf{A}_n^T]
\end{equation}
Suppose that we have a master node and $m$ worker nodes. Worker $i$ first stores $1/n$ fraction of matrix $A$, defined as $\tilde{\mathbf{A}}_i\in\mathbb{R}^{\frac{r}{n}\times t}$. Then it can compute a partial linear transform $\tilde{\mathbf{y}}_i=\tilde{\mathbf{A}}_i\mathbf{x}$ and return it to the master node. The master node waits only for the results of a subset of workers $\{\tilde{\mathbf{A}}_i\mathbf{x}|i\in I\subseteq[m]\}$ to recover the final output $\mathbf{y}$ using certain decoding algorithms. In general, we assume that the scale of data is much larger than the scale of the number of workers, i.e., $m\ll r,t$. The main framework is illustrated in Fig~\ref{fig:model}.

Given the above system model, we can formulate the coded distributed linear transform problem based on the following definitions.
\begin{definition}\label{def:coded_scheme}
	\emph{(Coded computation strategy)} A coded computation strategy is an $m\times n$ coding matrix $\mathbf{M}=[m_{ij}]_{i\in[m],j\in[n]}$ that is used to compute each $\tilde{\mathbf{A}}_i$. Specifically,
	\begin{equation}\label{eq:code_scheme}
	\tilde{\mathbf{A}}_i=\sum\limits_{j=1}^{n}m_{ij}\mathbf{A}_j, \forall i\in[m].
	\end{equation}
	Then each worker $i$ computes a partial linear transform $\tilde{\mathbf{y}}_i=\tilde{\mathbf{A}}_i\mathbf{x}$.
\end{definition}
This is a general definition of a large class of coded computation schemes. For example, in the polynomial code~\cite{yu2017polynomial}, the coding matrix $\mathbf{M}$ is the Vandermonde matrix. In the MDS type of codes~\cite{dutta2016short,lee2017speeding,lee2017coded}, it is a specific form of corresponding generator matrix.

\begin{definition}\label{def:recovery}
	\emph{(Recovery threshold)} A coded computation strategy $\mathbf{M}$ is $k$-recoverable if for any subset $I\subseteq [m]$ with $|I| = k$, the master node can recover $\mathbf{Ax}$ from $\{\tilde{\mathbf{y}_i}|i\in I\}$. The recovery threshold $\kappa(\mathbf{M})$ is defined as the minimum integer $k$ such that strategy $\mathbf{M}$ is $k$-recoverable.
\end{definition}

\begin{definition}\label{def:load}
	\emph{(Computation load)} The computation load of strategy $\mathbf{M}$ is defined as $l(\mathbf{M})=\|\mathbf{M}\|_{0}$, the number of the nonzero elements of coding matrix.
\end{definition}

In the task of computing a single linear transform, the computation load can be regarded as the total number of partial linear transforms $\mathbf{A}_i\mathbf{x}$ calculated locally. In other machine learning applications such as stochastic gradient descent, although one can encode the data matrix once and use it in subsequent iterations, the data set is extremely large-scale and sparse, i.e., $nnz(\mathbf{A})\ll rt$, and the density of coded matrix $\tilde{\mathbf{A}}_i$ is roughly equal to $\|\mathbf{M}_i\|_0$ times of the original data matrix $\mathbf{A}$. Therefore, $l(\mathbf{M})$ denotes the factor of total increasing computation time due to the coding operation in (\ref{eq:code_scheme}).

\begin{figure}[t]
	\vskip -0.05in
	\begin{center}
		\centerline{\includegraphics[width=5in]{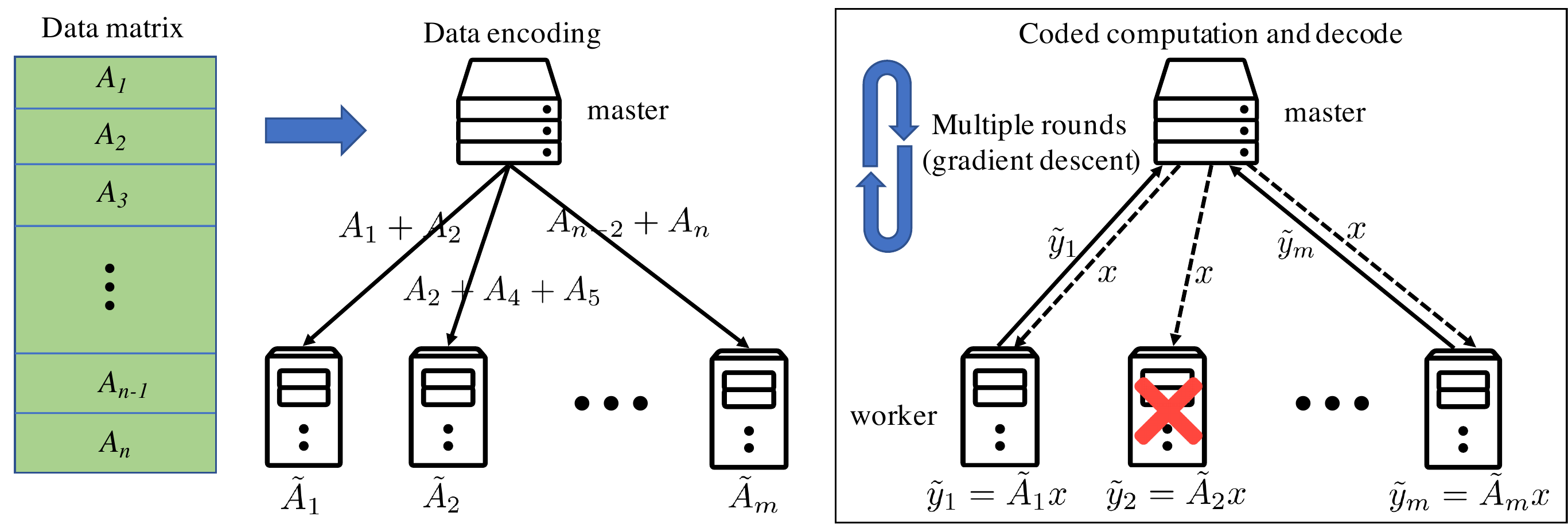}}
		\vskip 0in
		\caption{Framework of coded distributed linear transform}
		\vskip -0.4in
		\label{fig:model}
	\end{center}
\end{figure}


\section{Diagonal Code and its Optimality}

In this section, we first describe the optimum recovery threshold and optimum computation load. Then we introduce the $s$-diagonal code that exactly matches such lower bounds.

\subsection{Fundamental Limits on Recovery Threshold and Computation Load}
We first establish the information theoretical lower bound of the recovery threshold.
\begin{theorem}\label{thm:lowerbound_recovery}
	\emph{(Optimum recovery threshold)} For any coded computation scheme $\mathbf{M}\in\mathbb{R}^{m\times n}$ using $m$ workers that can each store $\frac{1}{n}$ fraction of $\mathbf{A}$, the minimum recovery threshold satisfies
	\begin{equation}
	\kappa^*=\min\limits_{\mathbf{M}\in\mathbb{R}^{m\times n}}\kappa(\mathbf{M})\geq n.
	\end{equation}
\end{theorem}
The proof is similar with the one in~\cite{yu2017polynomial} that uses the cut-set type argument around the master node. This bound implies that the optimum recovery threshold is independent of the number of workers. We next establish the lower bound of the computation load, i.e., the density of coding matrix $\mathbf{M}$.
\begin{theorem}\label{thm:lowerbound_load}
	\emph{(Optimum computation load)} For any coded computation scheme $\mathbf{M}\in\mathbb{R}^{m\times n}$ using $m$ workers that can each stores $\frac{1}{n}$ fraction of $\mathbf{A}$, to resist $s$ stragglers, we must have
	\begin{equation}
	l(\mathbf{M})\geq n(s+1).
	\end{equation}
\end{theorem}
The polynomial code~\cite{yu2017polynomial} achieves the optimum recover threshold of $n$. However, the coding matrix (Vandermonde matrix) is fully dense, i.e., $l(\mathbf{M})=nm$, and far beyond the above bound. The short MDS code~\cite{dutta2016short} can reduce the computation load but sacrifices the recovery threshold. Therefore, a natural question that arises is, can we design a code achieves both optimum recovery threshold and optimum computation load? We will answer this question in the sequel of this paper.

\subsection{Diagonal Code}

Now we present the $s$-diagonal code that achieves both the optimum recovery threshold and optimum computation load for any given parameter values of $n$, $m$ and $s$.
\begin{definition}\label{def:diagonal}
	\emph{($s$-diagonal code)} Given parameters $m$,$n$ and $s$, the $s$-diagonal code is defined as
	\begin{equation}
	\tilde{\mathbf{A}}_i=\sum\nolimits_{j=\max\{1,i-s\}}^{\min\{i,n\}}m_{ij}\mathbf{A}_j, \forall i\in[m],
	\end{equation}
	where each coefficient $m_{ij}$ is chosen from a finite set $S$ independently and uniformly at random.
\end{definition}
The reason we name this code as the $s$-diagonal code is that the nonzero positions of the coding matrix $\mathbf{M}$ show the following diagonal structure.

\begin{equation}
\mathbf{M}^T=\begin{bmatrix}
\bovermat{$s+1$}{* & * & \cdots & * & *& 0 }& 0 & 0 & \cdots & 0 \\
0 & * & * & \cdots & * & * & 0 & 0 & \cdots & 0 \\
0 & 0 & * & *  & \cdots & * & * & 0 & \cdots & 0 \\
\vdots & \vdots & \vdots & \vdots  & \ddots & \ddots & \ddots & \ddots  & \ddots  & \vdots \\
0 & \cdots & 0 & 0 & * & * & * & \cdots & * & 0 \\
0 & \cdots & 0 & 0 & 0 & * & * & \cdots & * & * \\
\end{bmatrix},
\end{equation}
where $*$ indicates nonzero entries of $\mathbf{M}$. Before we analyze the recovery threshold of the diagonal code, the following example provides an instantiation of the above construction for $n=4$, $s=1$ and $m=5$.

\textbf{Example 1}: ($1$-Diagonal code) Consider a distributed linear transform task $\mathbf{Ax}$ using $m=5$ workers. We evenly divide the matrix $\mathbf{A}$ along the row side into $4$ submatrices: $\mathbf{A}^T=[\mathbf{A}_1^T,\mathbf{A}_2^T,\mathbf{A}_3^T,\mathbf{A}_4^T]$. Given this notation, we need to compute the following $4$ uncoded components.
\begin{equation}
\mathbf{Ax}=\begin{bmatrix}
\mathbf{A}_1\mathbf{x}\\
\mathbf{A}_2\mathbf{x}\\
\mathbf{A}_3\mathbf{x}\\
\mathbf{A}_4\mathbf{x}
\end{bmatrix}
\end{equation}
Based on the definition of $1$-diagonal code, each worker stores the following submatrices,
\begin{equation}
\tilde{\mathbf{A}}_1=\mathbf{A}_1,\text{ }\tilde{\mathbf{A}}_2=\mathbf{A}_1+\mathbf{A}_2,\text{ }\tilde{\mathbf{A}}_3=\mathbf{A}_2+\mathbf{A}_3,\text{ }\tilde{\mathbf{A}}_4=\mathbf{A}_3+\mathbf{A}_4,\text{ }\tilde{\mathbf{A}}_5=\mathbf{A}_4.
\end{equation}
Suppose that the first worker is straggler and master node receives results from worker $\{2,3,4,5\}$. According to the above coded computation strategy, we have
\begin{equation}
\begin{bmatrix}
\tilde{\mathbf{y}}_2\\
\tilde{\mathbf{y}}_3\\
\tilde{\mathbf{y}}_4\\
\tilde{\mathbf{y}}_5
\end{bmatrix}=\begin{bmatrix}
1 & 1 & 0 & 0\\
0 & 1 & 1 & 0\\
0 & 0 & 1 & 1\\
0 & 0 & 0 & 1\\
\end{bmatrix}\begin{bmatrix}
\mathbf{A}_1\mathbf{x}\\
\mathbf{A}_2\mathbf{x}\\
\mathbf{A}_3\mathbf{x}\\
\mathbf{A}_4\mathbf{x}
\end{bmatrix}\quad\overset{\text{inverse}}{\Longrightarrow}\quad\begin{bmatrix}
\mathbf{A}_1\mathbf{x}\\
\mathbf{A}_2\mathbf{x}\\
\mathbf{A}_3\mathbf{x}\\
\mathbf{A}_4\mathbf{x}
\end{bmatrix}=\begin{bmatrix}
1 & 0 & 0 & 0\\
-1 & 1 & 0 & 0\\
0 & 0 & 1 & -1\\
0 & 0 & 0 & 1\\
\end{bmatrix}\begin{bmatrix}
\tilde{\mathbf{y}}_2\\
\tilde{\mathbf{y}}_3\\
\tilde{\mathbf{y}}_4\\
\tilde{\mathbf{y}}_5
\end{bmatrix}
\end{equation}
The coefficient matrix is an upper diagonal matrix, which is invertible since the elements in the main diagonal are nonzero. Then we can recover the uncoded components $\{\mathbf{A}_i\mathbf{x}\}$ by direct inversion of the above coefficient matrix. The decodability for the other $4$ possible scenarios can be proved similarly. Therefore, this code achieves the optimum recovery threshold of $4$. Besides, the total number of partial linear transforms is $8$, which also matches the lower bound of computation load.

\subsection{Optimality of Diagonal Code\label{sec:dc}}

The optimality of computation load of the $s$-diagonal code can be easily obtained by counting the number of nonzero elements of coding matrix $\mathbf{M}$. In terms of recovery threshold, we have the following result of the diagonal code
\begin{theorem}
	Let finite set $S$ satisfy $|S|\geq 2n^2C^{n}_m$, then the $s-$diagonal code achieves the recovery threshold of $n$. 
\end{theorem}
For each subset $U\subseteq[m]$ with $|U|=n$, let $\mathbf{M}^U$ be a $n\times n$ submatrices consisting rows of $\mathbf{M}$ index by $U$. To prove $s$-diagonal code achieves the optimum recovery threshold of $n$, we need to show all the $n\times n$ submatrices $\mathbf{M}^U$ are full rank. We first define the following bipartite graph model between the set of $m$ workers $[m]$ and the set of $n$ data partitions $[n]$.
\begin{definition}\label{def:dc_graph}
	Let $G^D(V_1,V_2)$ be a bipartite graph with $V_1=[m]$ and $V_2=[n]$. Each node $i\in V_2$ is connected to nodes $\{i,i+1,\ldots,i+s\}\subseteq V_1$. Define an Edmonds matrix $\mathbf{M}(\mathbf{x})\in\mathbb{R}^{m\times n}$ of graph $G^D(V_1,V_2)$ with $[\mathbf{M}(\mathbf{x})]_{ij}=x_{ij}$ if nodes $i\in V_1$ and $j\in V_2$ are connected; $[\mathbf{M}(x)]_{ij}=0$, otherwise. 
\end{definition}
Based on the above definition, the coding matrix $\mathbf{M}$ of $s$-diagonal code can be obtained by assigning each intermediate
$x_{ij}$ of Edmonds matrix $\mathbf{M}(\mathbf{x})$ a value from set $S$ independently and uniformly at random.  Given the subset $U\subseteq[m]$ with $|U|=n$, define $G^D(U,V_2)$ as a subgraph of $G^D(V_1,V_2)$ and $\mathbf{M}^U(\mathbf{x})$ as the corresponding Edmonds matrix. Then the probability that matrix $\mathbf{M}^U$ is full rank is equal to the probability that the determinant of the Edmonds matrix $\mathbf{M}^U(\mathbf{x})$ is nonzero at the given value $\mathbf{x}$. The following technical lemma~\cite{schwartz1980fast} provides a simple lower bound of the such an event.
\begin{lemma} \emph{(Schwartz-Zeppel Lemma)} Let $f(x_1,\ldots,x_{n^2})$ be a nonzero polynomial with degree $n^2$. Let $S$ be a finite set in $\mathbb{R}$. If we assign each variable a value from $S$ independently and uniformly at random, then
	\begin{equation}
	\mathbb{P}(f(x_1,x_2,\ldots,x_N)\neq 0)\geq 1- n^2/|S|.
	\end{equation}
\end{lemma}
A classic result in graph theory is that a bipartite graph $G^D(U,V_2)$ contains a perfect matching if and only if the determinant of Edmonds matrix, i.e., $|\mathbf{M}^U(\mathbf{x})|$, is a nonzero polynomial. Combining this result with Schwartz-Zeppel Lemma, we can finally reduce the analysis of the full rank probability of the submatrix $\mathbf{M}^U$ to the probability that the subgraph $G^D(U,V_2)$ contains a perfect matching.
\begin{equation}\label{pf:eq:rank_matching}
\mathbb{P}(|\mathbf{M}^U|\neq 0)=\underbrace{\mathbb{P}(|\mathbf{M}^U|\neq 0\big||\mathbf{M}^U(x)| \not\equiv​0)}_{\text{ S-Z Lemma: }\geq 1-1/2C^{n}_m}\cdot\underbrace{\mathbb{P}(|\mathbf{M}^U(x)|\not\equiv​0)}_{\text{$G^D(U,V_2)$ contains perfect matching}}
\end{equation}
The next technical lemma shows that for all subsets $U\subseteq [m]$, the subgraph $G^D(U,V_2)$ exactly contains a perfect matching. Then utilizing the union bound, we conclude that, with probability at least $1/2$, all the $n\times n$ submatrices of $\mathbf{M}$ are full rank. Since we can generate the coding matrix $\mathbf{M}$ offline, with a few
rounds of trials ($2$ in average), we can find a coding matrix with all $n\times n$ submatrices being full rank.
\begin{lemma}\label{lm:dc_matching}
	Construct bipartite graph $G^D(V_1,V_2)$ from Definition~\ref{def:dc_graph}. For each $U\subseteq [m]$ with $|U|=n$, the subgraph $G^D(U,V_2)$ contains a perfect matching.
\end{lemma}
The proof is mainly based on Hall's marriage theorem. One special case of $s$-diagonal code is that all the nonzero elements of matrix $\mathbf{M}$ can be equal to $1$, when we are required to resist only one straggler.
\begin{corollary}\label{coro:dc_code}
	Given the parameters $n$ and $m=n+1$, define the $1$-diagonal code: 
	\begin{equation*}
	\tilde{\mathbf{A}}_i=\left\{
	\begin{aligned}
	&\mathbf{A}_1,i=1; \mathbf{A}_n,i=n+1\\
	&\mathbf{A}_{i-1}+\mathbf{A}_i, 2\leq i \leq n
	\end{aligned}
	\right..
	\end{equation*}
	It achieves the optimum computation load of $2n$ and optimum recovery threshold of $n$.
\end{corollary}

The $s$-diagonal code achieves both optimum recovery threshold of $n$ and computation load of $n(s+1)$. This result implies that the the average computation load, i.e., $n(s+1)/(n+s)$, will increase when number of stragglers increases. In the extreme case, the coding matrix $\mathbf{M}$ will degenerate to a fully dense matrix. The reason behind this phenomenon mainly derives from the pessimistic definition of the recovery threshold, which requires the coded computation scheme $\mathbf{M}$ to resist any $s$ stragglers. In the next section, we show that, if we relax requirement in Definition~\ref{def:recovery} to the probabilistic scenario, we can design a code with a computation load independent of the number of stragglers $s$.

\section{Random Code: ``Break'' the Limits}

\subsection{Probabilistic Recovery Threshold}

In practice, the straggles randomly occur in each worker, and the probability that a specific straggling configuration happen is in low probability. We first demonstrate the main idea through the following motivating examples. 

\textbf{Example 2:} Consider a distributed linear transform tasks with $n=20$ data partitions, $s=5$ stragglers and $m=25$ workers. The  recovery threshold of $20$ implies that all $C_{25}^{20}=53130$ square $20\times 20$ submatrices are full rank. Suppose that a worker being a straggler is identically and independently Bernoulli random variable with probability $10\%$. Then, the probability that workers $\{1,2,3,4,5\}$ being straggler is $10^{-5}$. if there exists a scheme that can guarantee the master to decode the results in all configurations except straggling configurations $\{1,2,3,4,5\}$, we can argue that this scheme achieves a recovery threshold of $20$ with probability $1-10^{-5}$.  

\textbf{Example 3:} Consider the same scenario of \textbf{Example 1} ($n=4,s=1,m=5$). We change the coded computation strategy of the second worker from $\tilde{\mathbf{A}}_2=\mathbf{A}_1+\mathbf{A}_2$ to $\tilde{\mathbf{A}}_2=\mathbf{A}_2$. Based on the similar analysis, we can show that the new strategy can recover the final result from $4$ workers except the scenario that first worker being straggler. Therefore, the new strategy achieves recovery threshold of $4$ with probability $0.9$, and reduces the computation load by $1$.

Based on the above two examples, the computation load can be reduced when we allow the coded computation strategy fails in some specific scenarios. Formally, we have the following definition of the probabilistic recovery threshold.
\begin{definition}
	\emph{(Probabilistic recovery threshold)} A coded computation strategy $\mathbf{M}$ is probabilistic $k$-recoverable if for each subset $I\subseteq [m]$ with $|I| = \kappa(\mathbf{M})$, the master node can recover $\mathbf{Ax}$ from $\{\tilde{\mathbf{A}}_i\mathbf{x}|i\in I\}$ with high probability. The probabilistic recovery threshold $\kappa(\mathbf{M})$ is defined as the minimum integer $k$ such that strategy $\mathbf{M}$ is probabilistic $k$-recoverable.
\end{definition}
Instead of guaranteeing the recoverability of all the straggling configurations, new definition provides a probabilistic relaxation such that a small percentage of all the straggling configurations, i.e., vanishes as $n\rightarrow \infty$, are allowed to be unrecoverable. In the sequel, we show that, under such a relaxation, one can construct a coded computation scheme that achieves probabilistic recovery threshold of $n$ and a constant (regarding parameter $s$) computation load.

\subsection{Construction of Random Code}

Based on our previous analysis of the recovery threshold of $s-$diagonal code, we show that, for any subset $U\subseteq [m]$ with $|U|=n$, the probability that $\mathbf{M}^U$ is full rank is lower bounded by the probability (multiplied by $1-o(1)$) that the corresponding subgraph $G(U,V_2)$ contains a perfect matching. This technical path motivates us to utilize the random proposal graph to construct the coded computation scheme. The first one is the following $p$-Bernoulli code, which is constructed from the ER random bipartite graph model~\cite{erdos1964random}.
\begin{definition}
	\emph{($p$-Bernoulli code)} Given parameters $m,n$, construct the coding matrix $\mathbf{M}$ as follows:
	\begin{equation}
	m_{ij}=\left\{
	\begin{aligned}
	&t_{ij},\text{with probability }p\\
	&0,\text{with probability }1-p
	\end{aligned}
	\right..
	\end{equation}
	where $t_{ij}$ is picked independently and uniformly from the finite set $S$.
\end{definition}

\begin{theorem}\label{thm:ber_code}
	For any parameters $m,n$ and $s$, if $p=2\log(n)/n$, the $p$-Bernoulli code achieves the probabilistic recovery threshold of $n$.
\end{theorem}
This result implies that each worker of $p$-Bernoulli code requires accessing $2\log(n)$ submatrices in average, which is independent of number of stragglers. Note that the existing works in distributed functional computation~\cite{lee2017coded} proposes a random sparse code that also utilizes Bernoulli random variable to construct the coding matrix. There exist two key differences: (i) the elements of our matrix can be integer valued, while the random sparse code adopts real-valued matrix; (ii) the density of $p$-Bernoulli code is $2\log(n)$, while the density of the random sparse code is an unknown constant.

The second random code is the following $(d_1,d_2)$-cross code.
\begin{definition}
	\emph{($(d_1,d_2)$-cross code)} Given parameters $m,n$, construct the coding matrix $\mathbf{M}$ as follows: (1) Each row (column) independently and uniformly choose $d_1$ ($d_2$) nonzero positions; (2) For those nonzero positions, assign the value independently and uniformly from the finite set $S$.
\end{definition}
It is easy to see that the computation load of $(d_1,d_2)$-cross code is upper bounded by $d_1m+d_2n$. The next theorem shows that a constant choice of $d_1$ and $d_2$ can guarantee the probabilistic recovery threshold of $n$.
\begin{theorem}\label{thm:cross_code}
	For any parameters $m,n$ and $s$, if $s=\text{poly}(\log(n))$, the $(2,3)$-cross code achieves the probabilistic recovery threshold of $n$. If  $s=\Theta(n^{\alpha})$, $\alpha<1$, the $(2,2/(1-\alpha))$-cross code achieves the probabilistic recovery threshold of $n$.
\end{theorem}
The proof of this theorem is based on analyzing the existence of perfect matching in a random bipartite graph constructed as following: (i) each node in the left partition randomly and uniformly connects to $d_1$ nodes in the opposite class; (ii) each node in the right partition randomly and uniformly connects to $l$ nodes in the opposite class, where $l$ is chosen under a specific degree distribution. This random graph model can be regarded as a generalization of Walkup's $2$-out bipartite graph model~\cite{walkup1980matchings}. The main technical difficulty in this case derives from the intrinsic complicated statistical model of node degree of right partition.

\begin{remark}
	One straightforward construction is using the LDPC codes~\cite{luby2001efficient} or rateless codes such as LT and raptor code~\cite{luby2002lt}. These codes can not only guarantee the probabilistic recovery threshold of $\Theta(n)$, but also provide a low computation load, i.e., $O(n\log(1/\epsilon))$ from raptor code. However, the desired performance of these codes requires the $n$ sufficiently large, i.e., $n>10^5$, which is impossible in distributed computation. For example, when $n=20$, the recovery threshold of LT code  is roughly equal to $30$, which is much larger than the lower bound.
\end{remark}

\subsection{Numerical Results of Random Code}

\begin{figure}[t]
	\vskip -0.05in
	\begin{center}
		\centerline{\includegraphics[width=7.0in]{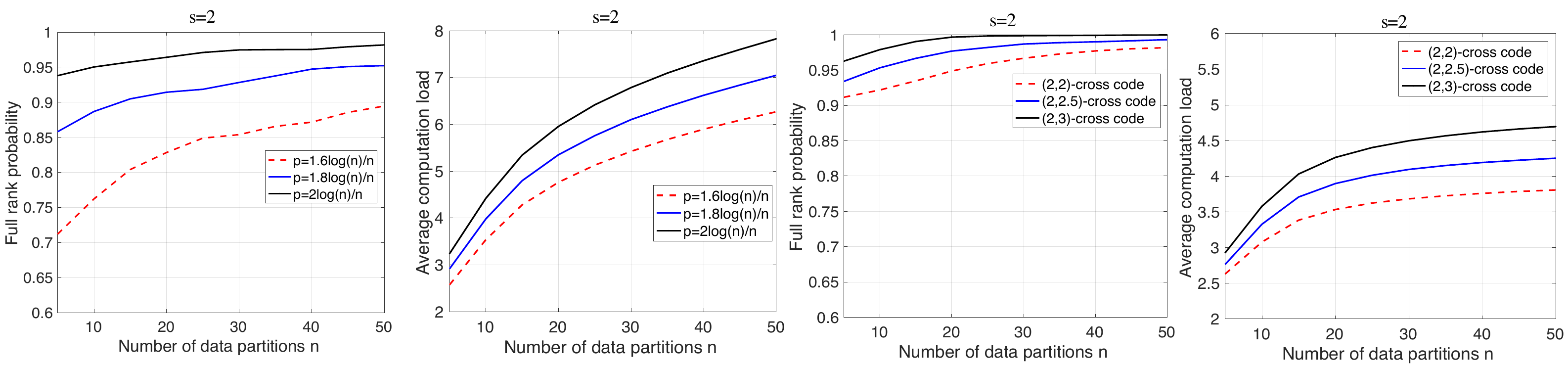}}
		\vskip -0.1in
		\caption{Statistical convergence speed of full rank probability and average computation load of random code under number of stragglers $s=2$.}
		\vskip -0.4 in
		\label{fig:numerical}
	\end{center}
\end{figure}

\begin{figure}[t]
	\vskip 0.05in
	\begin{center}
		\centerline{\includegraphics[width=7.0in]{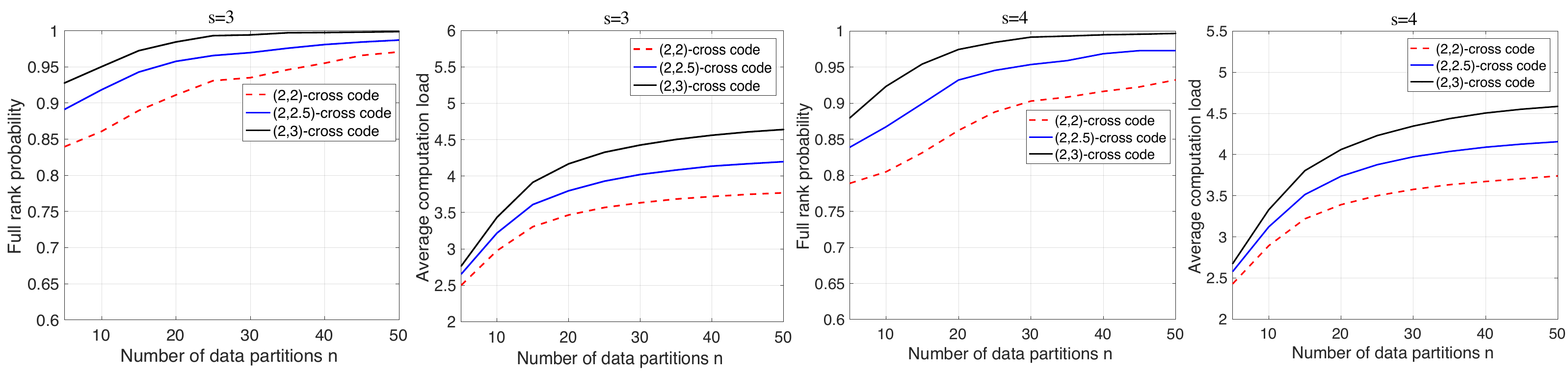}}
		\vskip -0.1in
		\caption{Statistical convergence speed of full rank probability and average computation load of $(d_1,d_2)$-cross code under number of stragglers $s=3,4$.}
		\vskip -0.4in
		\label{fig:numerical}
	\end{center}
\end{figure}

We examine the performance of proposed $p$-Bernoulli code and $(d_1,d_2)$-cross code in terms of convergence speed of full rank probability and computation load. In Fig.~\ref{fig:numerical}, we plot the percentage of full rank $n\times n$ square submatrix and the average computation load $l(\mathbf{M})/m$ of each scheme, based on $1000$ experimental run. Each column of the $(2,2.5)$-code  independently and randomly chooses $2$ or $3$ nonzero elements with equal probability. It can be observed that the full rank probability of both $p$-Bernoulli code and $(d_1,d_2)$-cross code converges to $1$ quite fast. The $(d_1,d_2)$-cross code exhibits even faster convergence and much less computation load compared to the $p$-Bernoulli code. For example, when $n=20,s=4$, the $(2,2)$-cross code achieves the full rank probability of $0.86$ and average computation load $3.4$. This provides the evidence that $(d_1,d_2)$-cross code is useful in practice. Moreover, in the practical use of random code, one can run multiple rounds of trails to find a ``best'' coding matrix with even higher full rank probability and lower computation load.

\section{Fast Decoding Algorithm}

The original decoding algorithm of both $s$-diagonal code and our constructed random code is based on inverting the coding matrix. Although the inverse of such a matrix can be found efficiently offline, the mapping from $[\mathbf{\tilde{y}}_{i_1},\mathbf{\tilde{y}}_{i_2},\ldots,\mathbf{\tilde{y}}_{i_n}]$ to result vector $\mathbf{y}$ incurs a complexity of $O(nr)$, which is large when number of data partitions $n$ and row dimension $r$ are large. We next show this can be further reduced by a hybrid decoding algorithm between the peeling decoding and Gaussian elimination. The main procedure is listed in Algorithm~\ref{alg:spcode}. The master node first finds a worker (ripple) that computes an uncoded tasks, i.e., $\mathbf{A}_i\mathbf{x}$, and recovers that block. Then for each collected results $\mathbf{\tilde{y}}_i$, it subtracts this block if the computation task $\mathbf{\tilde{A}}_i\mathbf{x}$ contains this block. If there exists no ripple in our decoding process, we go to a \textbf{rooting step}: randomly pick a particular block $\mathbf{A}_i\mathbf{x}$, and recover this block via a linear combination of the results $\{\mathbf{\tilde{y}}_i\}_{i=1}^{n}$.
\begin{lemma} \emph{(rooting step)}\label{lm:rooting}
	If rank$(\mathbf{M}^U)=n$, then for any $k_0\in\{1,2,\ldots,n\}$, we can recover a particular block $\mathbf{A}_i\mathbf{x}$ with column index $k_0$ in matrix $\mathbf{M}^U$ via the following linear combination.
	\begin{equation}\label{eq:borrow}
	\mathbf{A}_i\mathbf{x}=\sum\nolimits_{k=1}^n u_{k} \mathbf{\tilde{y}}_k.
	\end{equation}
	The vector $u=[u_1,\ldots,u_n]$ can be determined by solving $\mathbf{M}^{\intercal}\mathbf{u}=\mathbf{e}_{k_0}$, where $\mathbf{e}_{k_0}\in\mathbb{R}^{n}$ is a unit vector with unique $1$ locating at the index $k_0$. 
\end{lemma}

The basic intuition is to find a linear combination of row vectors of matrix $\mathbf{M}^U$ such that the row vectors eliminate all other blocks except the particular block $\mathbf{A}_i\mathbf{x}$.  Here we analyze the complexity of hybrid decoding algorithm. During each iteration, the complexity of operation $\mathbf{\tilde{y}}_{j}=\mathbf{\tilde{y}}_{j}-m_{ji}\mathbf{A}_i\mathbf{x}$ is $O(r/n)$,  and the complexity in each rooting step (\ref{eq:borrow}) is $O(r)$. Suppose there are number of $c$ blocks recovered from rooting step, $n-c$ blocks are recovered from peeling decoder. Further, each block $\mathbf{A}_i\mathbf{x}$ will be used $d_i$ times during the peeling decoding, where $d_i$ is the average density of column $\mathbf{M}_i$. Therefore, the total complexity is $O(r\bar{d}+rc(1-\bar{d}/n))$, where $\bar{d}=l(\mathbf{M})/n$. One can easily observe that the complexity of hybrid decoding algorithm is strictly less than the complexity of original inverse decoding $O(rn)$. In practice, we observe that number of recovery steps $c$ is constant for $n\leq10^3$, which implies that our hybrid decoding algorithm is nearly linear time for the above random codes.

\begin{algorithm}[tb]
	\caption{Hybrid decoding algorithm}
	\label{alg:spcode}
	\begin{algorithmic}
		\STATE Given the number of $n$ results with coefficient matrix $\mathbf{M}^U$.
		\REPEAT
		\STATE Find a row $\mathbf{M}^U_{i'}$ in matrix $\mathbf{M}^U$ with $\|\mathbf{M}^U_{i'}\|_0=1$.
		\IF{such row does not exist }
		\STATE Randomly pick a $i\in \{1,\ldots,n\}$ and recover corresponding block $\mathbf{A}_i\mathbf{x}$ by rooting step (\ref{eq:borrow}).
		\ELSE
		\STATE Recover the block $\mathbf{A}_i\mathbf{x}$ from $\mathbf{\tilde{y}}_{i'}$.
		\ENDIF
		\FOR{each computation results $\mathbf{\tilde{y}}_{j}$}
		\IF{$m_{ji}$ is nonzero}
		\STATE $\mathbf{\tilde{y}}_{j}=\mathbf{\tilde{y}}_{j}-m_{ji}\mathbf{A}_i\mathbf{x}$ and set $m_{ji}=0$.
		\ENDIF
		\ENDFOR
		\UNTIL{every block of vector $\mathbf{y}$ is recovered.}
	\end{algorithmic}
\end{algorithm} 

Moreover, we can show that, for the $s$-diagonal code,  when the position of blocks recovered from rooting step (\ref{eq:borrow}) is specifically designed, the Algorithm~\ref{alg:spcode} requires at most $s$ rooting steps. This result provides an $O(rs)$ time decoding algorithm for $s$-diagonal code.
\begin{theorem}\label{thm:decode_diag}
	\emph{(Nearly linear time decoding of $s$-diagonal code)}  For $s$-diagonal code constructed in Definition~\ref{def:diagonal} and any $n$ received results indexed by  $U=\{i_1,\ldots,i_n\}\subseteq [n+s]$ and $1\leq i_1\leq \cdots \leq i_n \leq n+s$. Let the index $k\in[n]$  be $i_k\leq n< i_{k+1}$. Recover the blocks indexed by $[n]\backslash\{i_1,\ldots,i_k\}$ from rooting steps (\ref{eq:borrow}). Then the Algorithm~\ref{alg:spcode} uses the at most $s$ times rooting steps and the total decoding time is $O(rs)$.
\end{theorem}

\section{Experimental Results}

In this section, we present the experimental results at OSC~\cite{OhioSupercomputerCenter1987}. We compare our proposed coding schemes including $s-$diagonal code and $(d_1,d_2)$-cross code against the following existing schemes in both single matrix vector multiplication and gradient descent: (i) \textbf{uncoded scheme}: the input matrix is divided uniformly across all workers without replication and the master waits for all workers to send their results; (ii) \textbf{sparse MDS code}~\cite{lee2017coded}: the generator matrix is a sparse random Bernoulli matrix with average computation overhead $\Theta(\log(n))$. (iii) \textbf{polynomial code}~\cite{yu2017polynomial}: coded matrix multiplication scheme with optimum recovery threshold and nearly linear decoding time; (iv) \textbf{short dot code}~\cite{dutta2016short}: append the dummy vectors to data matrix $\mathbf{A}$ before applying the MDS code, which provides some sparsity of encoded data matrix with cost of increased recovery threshold. (v) \textbf{LT code}~\cite{luby2002lt}: rateless code widely used in broadcast communication. It achieves average computation overhead $\Theta(\log(n))$ and a nearly linear decoding time using peeling decoder. To simulate straggler effects in large-scale system, we randomly pick $s$ workers that are running a background thread.

\subsection{Coded Linear Transform}

We implement all methods in python using MPI4py. Each worker stores the coded submatrix $\tilde{\mathbf{A}}_i$ according to the coding matrix $\mathbf{M}$. In the computation stage, each worker computes the linear transform $\tilde{\mathbf{A}}_i\mathbf{x}$ and returns the results using \texttt{Isend()}. Then the master node actively listens to the responses from each worker via \texttt{Irecv()}, and uses \texttt{Waitany()} to keep polling for the earliest finished tasks. Upon receiving enough results, the master stops listening and starts decoding the results. 

%

\begin{figure}[t]
	\vskip -0.05in
	\begin{center}
		\centerline{\includegraphics[width=7.0in]{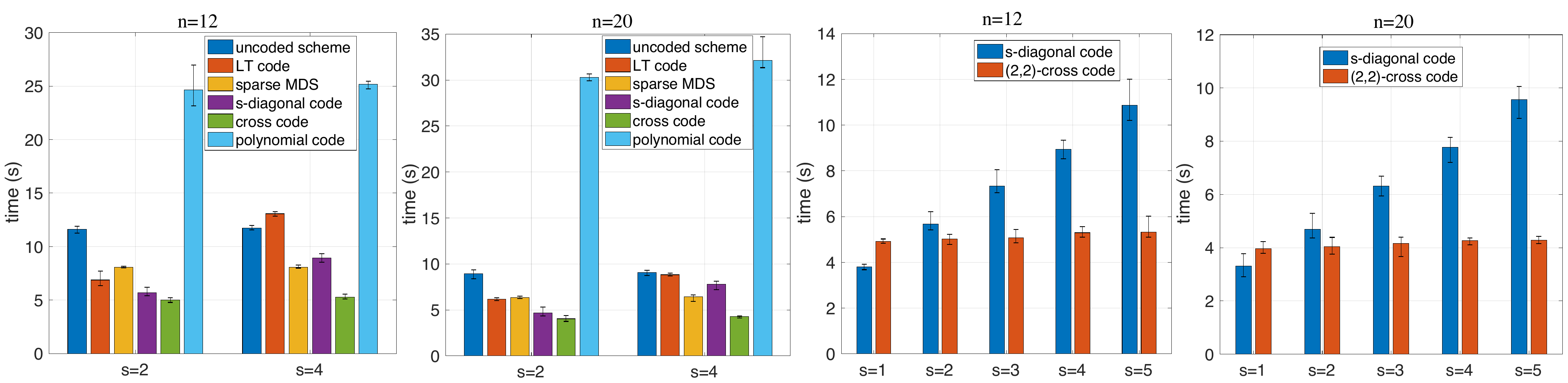}}
		\vskip 0in
		\caption{Comparison of job completion time including data transmission time, computation time and decoding time for $n=12,20$ and $s=2,4$.}
		\vskip -0.2in
		\label{fig:simres_Ax}
	\end{center}
\end{figure}

We first use a matrix with $r=t=1048576$ and nnz$(\mathbf{A})=89239674$ from data sets~\cite{davis2011university} , and evenly divide this matrix into $n=12$ and $20$ partitions. In Fig.~\ref{fig:simres_Ax}, we report the job completion time under $s=2$ and $s=4$, based on $20$ experimental runs. It can be observed that both $(2,2)$-cross code outperforms uncoded scheme (in 50\% the time), LT code (in 80\% the time), sparse MDS code (in 60\% the time), polynomial code (in 20\% the time) and our $s$-diagonal code. Moreover, we have the following observations: (i)LT code performs better than sparse MDS code when number of stragglers $s$ is small, and worse when $s$ increases. (ii) the uncoded scheme is faster than the polynomial code, because the density of encoded data matrix is greatly increased, which leads to increased computation time per worker and  additional I/O contention at the master node. 

We further compare our proposed $s$-diagonal code with $(2,2)$-cross code versus the number of stragglers $s$. As shown in Fig.~\ref{fig:simres_Ax}, when the number of stragglers $s$ increases, the job completion time of $s$-diagonal code increases while the $(2,2)$-cross code does not change. If the number of stragglers is smaller than $2$, the  $s$-diagonal code performs better than $(2,2)$-cross code. Another interesting observation is that the \emph{irregularity of work load} can decrease the  I/O contention. For example, when $s=2$, the computation load of $2$-diagonal code is similar as $(2,2)$-cross code, which is equal to $36$ in the case of $n=12$. However, the $(2,2)$-cross code costs less time due to the unbalanced worker load.

\subsection{Coded Gradient Descent}

\begin{figure}[t]
	\vskip -0.05in
	\begin{center}
		\centerline{\includegraphics[width=7.0in]{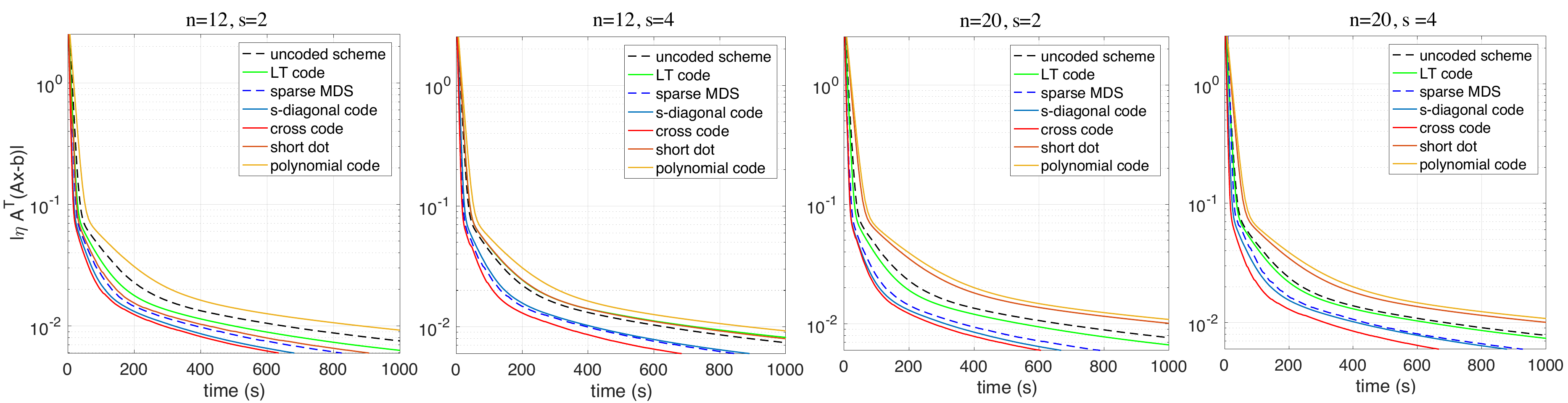}}
		\vskip 0in
		\caption{Magnitude of scaled gradient versus time for number of data partitions $n=12$, $20$ and number of stragglers $s=2$,$4$.}
		\vskip -0.4in
		\label{fig:simres_grad}
	\end{center}
\end{figure}

We first describe the gradient-based distributed algorithm to solve the following linear regression problem.
\begin{equation}
\min\limits_{\mathbf{x}}\frac{1}{2}\|\mathbf{Ax}-\mathbf{b}\|^2,
\end{equation}
where $\mathbf{A}\in\mathbb{R}^{r\times t}$ is the data matrix, $\mathbf{b}\in\mathbb{R}^{r}$ is the label vector and $\mathbf{x}\in\mathbb{R}^{t}$ is the unknown weight vector to be found. The standard gradient descent algorithm to solve the above problem is: in each iteration $t$,
\begin{equation}
\mathbf{x}_{t+1}=\mathbf{x}_t-\eta\mathbf{A}^{\intercal}(\mathbf{Ax}_t-\mathbf{b})=\mathbf{x}_t-\eta\sum\limits_{i=1}^{n}\mathbf{A}_i^{\intercal}\mathbf{A}_i\mathbf{x}_t+\eta\mathbf{A}^{\intercal}\mathbf{b}.
\end{equation}
In the uncoded gradient descent, each worker $i$ first stores a submatrix $\mathbf{A}_i$; then during iteration $t$, each worker $i$ first computes a vector $\mathbf{A}_i^{\intercal}\mathbf{A}_i\mathbf{x}_t$ and returns it to the master node; the master node updates the weight vector $\mathbf{x}_t$ according to the above gradient descent step and assigns the new weight vector $\mathbf{x}_{t+1}$ to each worker. The algorithm terminates until the gradient vanish, i.e., $\|\mathbf{A}^{\intercal}(\mathbf{Ax}_t-\mathbf{b})\|\leq \epsilon$. In the coded gradient descent, each worker $i$ first stores several submatrices according to the coding matrix $\mathbf{M}$; during iteration $t$, each worker $i$ computes a linear combinations,
\begin{equation}
\sum\limits_{j=1}^{n}m_{ij}\mathbf{A}_j^{\intercal}\mathbf{A}_j\mathbf{x}_t, \forall i\in[m].
\end{equation}
where $m_{ij}$ is the element of the coding matrix $\mathbf{M}$. The master node collects a subset of results, decodes the full gradient $\mathbf{A}^{\intercal}(\mathbf{Ax}_t-\mathbf{b})$ and updates the weight vector $\mathbf{x}_t$. Then continue to the next round.

We use a data from LIBSVM dataset repository with $r=19264097$ samples and $t=1163024$ features. We evenly divide the data matrix $\mathbf{A}$ into $n=12,20$ submatrices. In Fig.~\ref{fig:simres_grad} , we plot the magnitude of scaled gradient  $\|\eta\mathbf{A}^{\intercal}(\mathbf{Ax}-\mathbf{b})\|^2$ versus the running times of the above seven different schemes under $n=12,20$ and $s=2, 4$. Among all experiments, we can see that $(2,2)$-cross code converges at least $20\%$ faster than sparse MDS code, $2$ times faster than both uncoded scheme and LT code and at least $4$ times faster than short dot and  polynomial code. The $(2,2)$-cross code performs similar with $s$-diagonal code when $s=2$ and converges $30\%$ faster than $s$-diagonal code when $s=4$. 

\bibliographystyle{IEEEtran}
\bibliography{reff}

\appendix

\subsection{Proof of Theorem~\ref{thm:lowerbound_load}}

Given the parameter $m$ and $n$, considering any coded computation scheme that can resist $s$ stragglers. We first define the following bipartite graph model between the $m$ workers $[m]$ and $n$ data partitions $[n]$, where we connect node $i\in [m]$ and node $j\in[n]$ if $m_{ij}\neq 0$ (worker $i$ has access to data block $\mathbf{A}_j$). The degree of worker node $i\in [m]$ is $\|\mathbf{M}_i\|_0$. We show that the degree of  node $j\in[n]$ must be at least $s+1$. Suppose that it is less than $s+1$ and all its neighbors are stragglers. In this case, there exists no worker that is non-straggler and also access to $\mathbf{A}_j$(or the corresponding submatrix of coding matrix is rank deficient). It is contradictory to the assumption that it can resist $s$ stragglers.

Based on the above argument and the fact that, the sum of degrees of one partition is equal to the sum of degrees in another partition in the bipartite graph, we have
\begin{equation}
l(\mathbf{M})=\sum\limits_{i=1}^{n}\|\mathbf{M}_i\|_0\geq n(s+1).
\end{equation}
Therefore, the theorem follows.

\subsection{Proof of Lemma~\ref{lm:dc_matching}}

A direct application of Hall's theorem is: given a bipartite graph $G^D(V_1,V_2)$, for each $U\subseteq[m]$ with $|U|=n$, each subgraph $G^D(U,V_2)$ contains a prefect matching if and only if every subset set $S\subseteq V_1$ such that $|N(S)|<|S|$, where the neighboring set $N(S)$ is defined as $N(S)=\{y|x, y \text{ are connected for some } x\in S\}$. This result is equivalent to the following condition: for each subset $I\subseteq [m]$,
\begin{equation}
\left|\bigcup\limits_{i\in I}\text{supp}(\mathbf{M}_i)\right|\geq |I|,
\end{equation}
where supp$(\mathbf{M_i})$ is defined as the support set: supp$(\mathbf{M_i})=\{j|m_{ij}\neq 0,j\in[n]\}$, $\mathbf{M_i}$ is $i$th row of the coding matrix $\mathbf{M}$. Suppose that the set $I=\{i_1,i_2,\ldots,i_k\}$ with $i_1 < i_2<\ldots < i_k$ and $\text{supp}(\mathbf{M}_{i_{k_1}})\cap \text{supp}(\mathbf{M}_{i_{k_2}})\neq \emptyset$. Otherwise, we can divide the set into two parts $I_l=\{i_1,i_2,\ldots,i_{k_1}\}$ and $I_r=\{i_{k_2},i_2,\ldots,i_k\}$ and prove the similar result in these two sets. Based on our construction of diagonal code, we have 
\begin{equation}
\left|\bigcup\limits_{i\in I}\text{supp}(\mathbf{M}_i)\right|=\min\{i_k,n\}-\max\{1,i_1-s\}\overset{(a)}{\geq} k,
\end{equation}
The above, step (a) is based on the fact that $i_k-i_1\geq k$. Therefore, the lemma follows.

\subsection{Proof of Corollary~\ref{coro:dc_code}}
To prove $1$-diagonal code achieves the recovery threshold of $n$, we need to show that, for each subset $U\subseteq [n+1]$ with $|U|=n$, submatrix $\mathbf{M}^U$ is full rank. Let $U=[n+1]\backslash\{k\}$, the submatrix $\mathbf{M}^U$ satisfies
\begin{equation}
\mathbf{M}^U=\begin{bmatrix}
\mathbf{E}& \mathbf{0}\\
\mathbf{0} & \mathbf{F},
\end{bmatrix}
\end{equation}
where $\mathbf{E}$ is a $(k-1)$ dimensional square submatrix consists of first $(k-1)$ rows and columns, and $\mathbf{F}$ is a $(n-k+1)$ dimensional square submatrix consists of the last $(n-k+1)$ rows and columns. The matrix $\mathbf{E}$ is a lower diagonal matrix due to the fact that, for $i<j$,
\begin{equation}
E_{ij}=m_{ij}=0.
\end{equation}
The matrix $\mathbf{E}$ is a upper diagonal matrix due to the fact that, for $i>j$,
\begin{equation}
F_{ij}=m_{i+k,j+k-1}\overset{(a)}{=}0.
\end{equation}
The above, (a) utilizes the fact that $(i+k)-(j+k-1)\geq 2$ when $i>j$. Based on the above analysis, we have
\begin{equation}
\det(\mathbf{M}^U)=\det(\mathbf{E})\det(\mathbf{F})=1,
\end{equation}
which implies that matrix $\mathbf{M}^U$ is full rank. Therefore, the corollary follows.

\subsection{Proof of Theorem~\ref{thm:ber_code}}

Based on our analysis in Section~\ref{sec:dc}, the full rank probability of an $n\times n$ submatrix $\mathbf{M}^U$ can be lower bounded by a constant multiplying the probability of the existence of perfect matching in a bipartite graph. 
\begin{equation}\label{pf:eq:rank_match}
\mathbb{P}(|\mathbf{M}^U|\neq 0)=\underbrace{\mathbb{P}(|\mathbf{M}^U|\neq 0\big||\mathbf{M}^U(x)| \not\equiv​0)}_{\text{ S-Z Lemma: }\geq 1-1/4^m}\cdot\underbrace{\mathbb{P}(|\mathbf{M}^U(x)|\not\equiv​0)}_{\text{bipartite graph contains perfect matching}}
\end{equation}
Therefore, to prove that the $p$-Bernoulli code achieves the probabilistic recovery threshold of $n$, we need to show that each subgraph contains a perfect matching with high probability. Without loss of generality, we can define the following random bipartite graph model.
\begin{definition}\label{def:ber_graph}
	\emph{($p$-Bernoulli random graph)} Graph $G^b(U,V_2,p)$ initially contains the isolated nodes with $|U|=|V_2|=n$. Each node $v_1\in U$ and node $v_2\in V_2$ is connected with probability $p$ independently.
\end{definition}
Clearly, the above model describes the support structure of each submatrix $\mathbf{M}^U$ of $p$-Bernoulli code. The rest is to show that, with specific choice of $p$, the subgraph $G^b(U,V_2,p)$ contains a perfect matching with high probability.

The technical idea is to use the Hall's theorem. Assume that the bipartite graph $G^b(U,V_2, p)$ does not have a perfect matching. Then by Hall's condition, there exists a violating set $S\subseteq U$ or $S\subseteq V_2$ such that $|N(S)|<|S|$. Formally, by choosing such $S$ of smallest cardinality, one immediate consequence is the following technical statement.
\begin{lemma}\label{lm:halltheorem}
	If the bipartite graph $G^b(U,V_2, p)$ does not contain a perfect matching, then there exists a set $S\subseteq U$ or $S\subseteq V_2$ with the following properties.
	\begin{enumerate}
		\item $|S|=|N(S)|+1$.
		\item For each node $t\in N(S)$, there exists at least two adjacent nodes in $S$.
		\item $|S|\leq n/2$.
	\end{enumerate}
\end{lemma}
\textbf{Case 1:} We consider $S\subseteq U$ and $|S|=1$. In this case, we have $|N(S)|=0$ and need to estimate the probability that there exists one isolated node in partition $U$. Let random variable $X_i$ be the indicator function of the event that node $v_i$ is isolated. Then we have the probability that
\begin{equation*}
\mathbb{P}(X_i=1)=\left(1-p\right)^n,
\end{equation*}
Let $X$ be the total number of isolated nodes in partition $U$. Then we have
\begin{align}
\mathbb{E}[X] = {E}\left[\sum\limits_{i=1}^{n}X_i\right]=n\left(1-p\right)^n\overset{(a)}{\leq} \frac{1}{n}.
\end{align}
The above, step (a) utilizes the assumption that $p=2\log(n)/n$ and the inequality that $(1+x/n)^n\leq e^x$.

\textbf{Case 2:} We consider $S\subseteq U$ and $2\leq |S|\leq n/2$. Let $E$ be the event that such an $S$ exists, we have
\begin{align}
\mathbb{P(E)}\leq& \sum\limits_{k=2}^{n/2}\binom{n}{k}\binom{n}{k-1}\binom{k}{2}^{k-1}(1-p)^{k(n-k+1)}p^{2(k-1)}\notag\\
\overset{(a)}{<}&\sum\limits_{k=2}^{n/2}\frac{1}{6}\cdot\frac{n}{(n-k)(n-k+1)}\cdot\frac{n^{2n}}{k^{2k}(n-k)^{2(n-k)}}\cdot\left[\frac{k(k-1)}{2}\right]^{k-1}\cdot(1-p)^{k(n-k+1)}p^{2(k-1)}\notag\\
\overset{(b)}{<}&\sum\limits_{k=2}^{n/2}\frac{e^2n}{6k^2(n-k)(n-k+1)}\cdot\left(\frac{2\log^2(n)}{n}\right)^{k-1}
\notag\\
\overset{(c)}{<}&\sum\limits_{k=2}^{n/2}\frac{2}{3(n-1)}\cdot\left(\frac{2\log^2(n)}{n}\right)^{k-1}\notag\\
<&\frac{\log^2(n)}{3n}.
\end{align}
The above, step (a) is based on the inequality
\begin{equation}\label{eq:sterlin}
\sqrt{2\pi n}\left(\frac{n}{e}\right)^n \leq n!\leq \frac{60}{59}\sqrt{2\pi n}\left(\frac{n}{e}\right)^n, \forall n\geq 5.
\end{equation}
The step (b) utilizes the fact that $p=2\log(n)/n$, $k\leq n/2$ and the inequality $(1+x/n)^n\leq e^x$; step (c) is based on the fact that $k(n-k+1)\geq 2(n-1)$ and $k(n-k)\geq 2(n-2)$, $n/(n-2)<5/3$ for $k\geq 2$ and $n\geq 5$. Utilizing the union bound to sum the results in case 1 and case 2, we can obtain that the probability that graph $G(U,V_2,p)$ contains a perfect matching is at least
\begin{equation}
1-\frac{\log^2(n)}{3n}.
\end{equation}
Therefore, incorporating this result into estimating (\ref{pf:eq:rank_matching}), the theorem follows.

\subsection{Proof of Theorem~\ref{thm:cross_code}}

To prove the $(d_1,d_2)$-cross code achieves the probabilistic recovery threshold of $n$, we need to show that each subgraph of the following random bipartite graph contains a perfect matching with high probability.
\begin{definition}
	\emph{($(d_1,d_2)$-regular random graph)} Graph $G^c(V_1,V_2,d_1,d_2)$ initially contains the isolated nodes with $|V_1|=m$ and $|V_2|=n$. Each node $v_1\in V_1$ ($v_2\in V_2$) randomly and uniformly connects to $d_1$ ($d_2$) nodes in $V_1$ ($V_2$). 
\end{definition}
The corresponding subgraph is defined as follows.
\begin{definition}
	For each $U\subseteq V_1$ with $|U|=n$, the subgraph $G^c(U,V_2,d_1,\bar{d})$ is obtained by deleting the nodes in $V_1\backslash U$ and corresponding arcs.
\end{definition}
Clearly, the above definitions of $(d_1,d_2)$-regular graph and corresponding subgraph describe the support structure of the coding matrix $\mathbf{M}$ and submatrix $\mathbf{M}^U$ of the $(d_1,d_2)$-cross code. Moreover, we have the following result regarding the structure of each subgraph $G^c(U,V_2,d_1,\bar{d})$

\textbf{Claim.} For each $U\subseteq V_1$ with $|U|=n$, the subgraph $G^c(U,V_2,d_1,\bar{d})$ can be constructed from the following procedure" (i) Initially, graph $G^c(U,V_2,d_1,\bar{d})$  contain the isolated nodes with $|U|=|V_2|=n$; (ii) Each node $v_1\in U$ randomly and uniformly connects to $d_1$ nodes in $V_2$; (iii) Each node $v_2\in V_2$ randomly and uniformly connects to $l$ nodes in $V_1$, where $l$ is chosen according to the distribution:
\begin{equation}\label{eq:cross_code_deg}
\mathbb{P}(l)=\binom{n}{l}\binom{m-n}{d_2-l}\bigg/\binom{m}{d_2}, 0\leq l\leq d_2.
\end{equation}

Then, the rest is to show that the subgraph $G^c(U,V_2,d_1,\bar{d})$ contains a perfect matching with high probability.
\begin{definition}\label{def:blocking}
	\emph{(Forbidden $k$-pair)} For a bipartite graph $G(U,V_2,d_1,\bar{d})$, a pair $(A,B)$ is called a $k$-blocking pair if $A\subseteq U$ with $|A|=k$, $B\subseteq V_2$ with $|B|=n-k+1$, and there exists no arc between the nodes of sets $A$ and $B$. A blocking $k$-pair $(A,B)$ is called a forbidden pair if at least one of the following holds:
	\begin{enumerate}
		\item $2\leq k< (n+1)/2$, and for any $v_1\in A$ and $v_2\in V_2\backslash B$, $(A\backslash\{v_1\}, B\cup\{v_2\})$ is not a $(k-1)$-blocking pair.
		\item $(n+1)/2\leq k\leq n-1$, and for any $v_1\in U\backslash A$ and $v_2\in B$, $(A\cup\{v_1\}, B\backslash\{v_2\})$ is not a $(k+1)$-blocking pair.
	\end{enumerate}
\end{definition}

The following technical lemma modified from~\cite{walkup1980matchings} is useful in our proof.
\begin{lemma}\label{lm:fb_match}
	If the graph $G^c(U,V_2,d_1,\bar{d})$ does not contain a perfect matching, then there exists a forbidden $k$-pair for some $k$.
\end{lemma}
\begin{proof}
	One direct application of the Konig's theorem to bipartite graph shows that $G^c(U,V_2,d_1,\bar{d})$ contains a perfect matching if and only if it does not contain any blocking $k$-pair. It is rest to show that the existence of a $k$-blocking pair implies that there exists a forbidden $l-$pair for some $l$. Suppose that there exists a  $k$-blocking pair $(A,B)$ with $k< (n+1)/2$, and it is not a forbidden $k$-pair. Otherwise, we already find a forbidden pair. Then, it implies that there exists $v_1\in A$ and $v_2\in V_2\backslash B$ such that $(A\backslash\{v_1\}, B\cup\{v_2\})$ is a $(k-1)$-blocking pair. Similarly, we can continue above argument on blocking pair $(A\backslash\{v_1\}, B\cup\{v_2\})$ until we find a forbidden pair. Otherwise, we will find a $1-$blocking pair $(A',B')$, which is a contradiction to our assumption that each node $v_1\in U$ connects $d_1$ nodes in $V_2$. The proof for $k\geq(n+1)/2$ is same.
\end{proof}
Let $E$ be the event that graph $G(U,V_2,d_1,\bar{d})\text{ contains perfect matching}$. Based on the the results of Lemma~\ref{lm:fb_match}, we have
\begin{align}
&1-\mathbb{P}(E)=\mathbb{P}\left(\bigcup\limits_{k=2}^{n-1}k\text{-forbidden pair exists}\right)\notag\\
&\leq \sum\limits_{k=2}^{n-1}\mathbb{P}\left(k\text{-forbidden pair exists}\right)\notag\\
&\leq \sum\limits_{k=2}^{n-1}\binom{n}{k}\binom{n}{n-k+1}\cdot\mathbb{P}\left((A,B)\text{ is }k\text{-forbidden pair}\right)\notag\\
&=\sum\limits_{k=2}^{n-1}\binom{n}{k}\binom{n}{n-k+1}\alpha(k)\beta(k).\notag
\end{align}
The above, $A$ and $B$ are defined as node sets such that $A\subseteq U$ with $|A|=k$ and  $B\subseteq V_2$ with $|B|=n-k+1$. The $\alpha(k)$ and $\beta(k)$ are defined as follows.
\begin{align}
&\alpha(k)=\mathbb{P}\left((A,B)\text{ is }k\text{-forbidden pair}\big|(A,B)\text{ is }k\text{-blocking pair}\right),\\
&\beta(k)=\mathbb{P}\left((A,B)\text{ is }k\text{-blocking pair}\right).
\end{align}
From the Definition~\ref{def:blocking}, it can be obtained the following estimation of probability $\beta(k)$.
\begin{equation}\label{eq:org_beta}
\beta(k)=\left[\binom{k-1}{d_1}\bigg/\binom{n}{d_1}\right]^k\cdot\left[\sum\limits_{l=0}^{d_2}\binom{n-k}{l}\binom{m-n}{d_2-l}\bigg/\binom{m}{d_2}\right]^{n-k+1}.
\end{equation}
The first factor gives the probability that there exists no arc from nodes of $A$ to nodes of $B$. The second factor gives the probability that there exists no arc from nodes of $B$ to nodes of $A$. The summation operation in the second factor comes from conditioning such probability on the distribution (\ref{eq:cross_code_deg}). Based on the Chu-Vandermonde identity, one can simplify $\beta(k)$ as
\begin{equation}\label{eq:beta}
\beta(k)=\left[\binom{k-1}{d_1}\bigg/\binom{n}{d_1}\right]^k\cdot\left[\binom{m-k}{d_2}\bigg/\binom{m}{d_2}\right]^{n-k+1}.
\end{equation}
Utilizing the inequality
\begin{equation}\label{eq:sterlin_2}
\sqrt{2\pi n}\left(\frac{n}{e}\right)^n \leq n!\leq e^{\frac{1}{12n}}\sqrt{2\pi n}\left(\frac{n}{e}\right)^n, 
\end{equation}
we have
\begin{align}
&\binom{n}{k}\binom{n}{n-k+1}\leq \frac{n^{2n}}{k^{2k}(n-k)^{2(n-k)}}\cdot\frac{ne^{1/6n}}{2\pi(n-k)(n-k+1)},\label{pf:eq:comp1}
\end{align}
\begin{align}
\binom{k-1}{d_1}\bigg/\binom{n}{d_1} &\leq c_1 \sqrt{\frac{(k-1)(n-d_1)}{(k-d_1-1)n}}\left(\frac{k-1}{k-d_1-1}\right)^{k-d_1-1}\left(\frac{n-d_1}{n}\right)^{n-d_1}\left(\frac{k-1}{n}\right)^{d_1}\notag\\
&\overset{(a)}{\leq} c_1 \sqrt{\frac{k-1}{k-d_1-1}}\left(\frac{n-d_1}{n}\right)^{\frac{1}{2}-d_1}\left(\frac{k-1}{n}\right)^{d_1}.\label{pf:eq:comp2}
\end{align}
\begin{align}
\binom{m-k}{d_2}\bigg/\binom{m}{d_2}
&\overset{(a)}{\leq} c_2 \sqrt{\frac{m-k}{m-d_2-k}}\left(\frac{m-d_2}{m}\right)^{\frac{1}{2}-d_2}\left(\frac{m-k}{m}\right)^{d_2}.\label{pf:eq:comp3}
\end{align}
In the above, step (a) is based on fact that $(1+x/n)^n\leq e^x$ and parameters $c_1$ and $c_2$ are defined as
\begin{align}\label{pf:eq:comp4}
c_1 = e^{1/12(k-1)}e^{1/12(n-d_1)}, \text{ } c_2 = e^{1/12(m-k)}e^{1/12(m-d_2)}.
\end{align}
Combining the equations (\ref{pf:eq:comp1})-(\ref{pf:eq:comp4}), we can obtain that
\begin{align}
\gamma(k)&=\binom{n}{k}\binom{n}{n-k+1}\beta(k)\notag\\
&< c_3\left(1-\frac{1}{k}\right)^{2k}\cdot\frac{n(m-k)^{d_2}}{m^{d_2}(n-k)(n-k+1)}\cdot\left[\frac{n^2(m-k)^{d_2}}{(n-k)^2m^{d_2}}\right]^{n-k}\cdot\left(\frac{k-1}{n}\right)^{(d_1-2)k}.
\end{align}
The constant $c_3$ is given by $c_3=e^{d_1^2+13d_1/12+d_2+1/3}/2\pi$. The third term satisfies
\begin{equation}\label{pf:eq:keyest}
\frac{n^2(m-k)^{d_2}}{(n-k)^2m^{d_2}}\leq \max\left\{1,\frac{n^2(m-n+1)^{d_2}}{m^{d_2}}\right\}, \forall 2\leq k\leq n-1,
\end{equation}
which is based on the fact that, if $d_2>2$, the function
\begin{equation*}
f(k)=\frac{n^2(m-k)^{d_2}}{(n-k)^2m^{d_2}} 
\end{equation*}
is monotonically decreasing when $k\leq (d_2n-2m)/(d_2-2)$ and  increasing when $k\geq (d_2n-2m)/(d_2-2)$. If $d_2=2$, it is monotonically increasing for $k\geq 0$.

We then estimate the conditional probability $\alpha(k)$. Given a blocking pair $A\subseteq U$ with $|A|=k$ and $B\subseteq V_2$ with $|B|=n-k+1$, and a node $v_i\in A$, let $E_i$ be the set of nodes in $V_2\backslash B$ on which $d_1$ arcs from node $v_i$ terminate. Let $E'$ be the set of nodes $v$ in $V_2\backslash B$ such that at least $2$ arcs leaving from $v$ to nodes in $A$. Then we have the following technical lemma.
\begin{lemma}\label{lm:estforbidden}
	Given a blocking pair $A\subseteq U$ with $|A|=k$ and $B\subseteq V_2$ with $|B|=n-k+1$, if $(A,B)$ is $k$-forbidden pair, then
	\begin{equation*}
	E^*=\left(\bigcup\limits_{i=1}^kE_i\right)\cup E' = V_2\backslash B.
	\end{equation*}
\end{lemma}
\begin{proof}
	Suppose that there exists node $v\in V_2\backslash (E^*\cup B)$, then there exists no arc from $A$ to $v$ and there exists at most $1$ arc from $v$ to $A$. If such an arc exists, let $v'$ be the corresponding terminating node in $A$. Then we have $(A\backslash \{v'\}, B\cup\{v\})$ is a blocking pair, which is contradictory to the definition of forbidden pair.  If such an arc does not exist, let $v'$ be the an arbitrary node in $A$. Then we have $(A\backslash \{v'\}, B\cup\{v\})$ is a blocking pair, which is also contradictory to the definition of forbidden pair.
\end{proof}
The lemma~\ref{lm:estforbidden} implies that we can upper bound the conditional probability by
\begin{equation}
\alpha(k)\leq \mathbb{P}\left[\left(\bigcup\limits_{i=1}^kE_i\right)\cup E'=V_2\backslash B\right]= (1- P_1^kP_2)^{k-1},
\end{equation}
where $P_1$ and $P_2$ is defined as: for any node $v\in V_2\backslash B$,
\begin{align}
P_1=\mathbb{P}(v \notin E_i)&=\binom{k-2}{d_1}\bigg/\binom{k-1}{d_1}=\frac{k-d_1-1}{k-1},\notag\\
P_2=\mathbb{P}(v \notin E')&=1-\sum\limits_{l_1=2}^{d_2}\sum\limits_{l_2=l_1}^{d_2}\mathbb{P}(l_2)\cdot\binom{k}{l_1}\binom{n-k}{l_2-l_1}\bigg/\binom{n}{l_2}\notag\\
&\overset{(a)}{=}\left[\binom{m-k}{d_2}+k\binom{m-k}{d_2-1}\right]\bigg/\binom{m}{d_2}\notag\\
&\overset{(b)}{>}e^{1/6}\left(\frac{m-k}{m}\right)^{m-k}\left(\frac{m-d_2}{m-d_2-k}\right)^{m-k-d_2}\left(\frac{m-d_2}{m}\right)^{k}\notag\\
&\overset{(c)}{>}c_4e^{1/6-d_2}.
\end{align}
The above, step (a) utilizes Chu-Vandermonde identity twice; step (b) is adopts the inequality (\ref{eq:sterlin_2}); step (c) is based on the fact that if $n$ is sufficiently large, $(1-x/n)^n\geq c_5e^{-x}$, where $c_5$ is a constant. Combining the above estimation of $P_1$ and $P_2$, we have the following upper bound of $\alpha(k)$.
\begin{equation}\label{pf:eq:est_alpha}
\alpha(k)<\left[1-c_6\left(\frac{k-d_1-1}{k-1}\right)^k\right]^{k-1}.
\end{equation}

We finally estimate the probability that the graph $G^c(U,V_2,d_1,\bar{d})$ contains a perfect matching under the following two cases.

\textbf{Case 1:} The number of stragglers $s=\text{poly}(\log(n))$. Let $d_1=2,d_2=3$. Based on the estimation (\ref{pf:eq:keyest}), we have
\begin{equation}
\frac{n^2(m-k)^{3}}{(n-k)^2m^{3}}\leq \max\left\{1,\frac{n^2(s+1)^{3}}{(n+s)^{3}}\right\}\leq 1, \text{ for } n \text{ sufficiently large.} 
\end{equation}
Combining the above results with the estimation of $\beta(k)$, we have
\begin{equation}
\gamma(k)\leq \frac{c_3e^{-2}}{n}, 2\leq k\leq n-1.
\end{equation}
Then we can obtain that
\begin{align}
\mathbb{P}(G(U,V_2,d_1,\bar{d})\text{ contains perfect matching})&\geq 1- \sum\limits_{k=1}^n\binom{n}{k}\binom{n}{n-k+1}\alpha(k)\beta(k)\notag\\
&\geq 1- \sum\limits_{k=2}^{n-1} \frac{c_3e^{-2}}{n} \alpha(k)\notag\\
&\overset{(a)}{>} 1- \frac{c_3e^{-2}}{n} \sum\limits_{k=2}^{n-1}  \left[1-c_6\left(\frac{k-3}{k-1}\right)^k\right]^{k-1}\notag\\
&\overset{(b)}{>} 1- \frac{c_7}{n}.
\end{align}
The above, step (a) utilizes the estimation of $\alpha(k)$ in (\ref{pf:eq:est_alpha}); step (b) is based on estimating the tail of the summation as geometric series $(1-c_6/e^2)^{k-1}$.

\textbf{Case 2:} The number of stragglers $s=\Theta(n^{\alpha})$, $\alpha<1$. Let $d_1=2,d_2=2/(1-\alpha)$. For $2\leq k\leq n-2$, we have
\begin{equation}
\frac{n^2(m-k)^{3}}{(n-k)^2m^{3}}\leq \max\left\{1,\frac{n^2(n^{\alpha}+2)^{2/(1-\alpha)}}{4(n+n^{\alpha})^{2/(1-\alpha)}}\right\}\leq 1, \text{ for } n \text{ sufficiently large.} 
\end{equation}
Combining the above results with the estimation of $\beta(k)$, we have
\begin{equation}
\gamma(k)\leq \frac{c_3e^{-2}}{n}, 2\leq k\leq n-2.
\end{equation}
For $k=n-1$, we have
\begin{equation}
\frac{n^2(m-k)^{3}}{(n-k)^2m^{3}}\leq \max\left\{1,\left(\frac{n+n^{1-\alpha}}{n+n^{\alpha}}\right)^{2/(1-\alpha)}\right\}.
\end{equation}
If $\alpha\geq1/2$, we can directly obtain that $\gamma(k)\leq \frac{c_3e^{-2}}{n}$. If $\alpha<1/2$, we have
\begin{equation}
\gamma(k)\leq \frac{c_3e^{-2}}{n}\left(\frac{n+n^{1-\alpha}}{n+n^{\alpha}}\right)^{4/(1-\alpha)}\leq \frac{c_8e^{-2}}{n}, \text{ for } n \text{ sufficiently large.} 
\end{equation}
Similarly, we can obtain that
\begin{align}
\mathbb{P}(G(U,V_2,d_1,\bar{d})\text{ contains perfect matching})&\geq 1- \sum\limits_{k=1}^n\binom{n}{k}\binom{n}{n-k+1}\alpha(k)\beta(k)\notag\\
&\geq 1- \sum\limits_{k=2}^{n-1} \frac{e^{-2}\max\{c_3,c_8\}}{n} \alpha(k)\notag\\
&\overset{(b)}{>} 1- \frac{c_9}{n}.
\end{align}
Therefore,  in both cases, incorporating the above results into estimating (\ref{pf:eq:rank_matching}), the theorem follows.

\subsection{Proof of Theorem~\ref{thm:decode_diag}}

Based on our construction, the cardinality of the set 
\begin{equation}
|[n]\backslash\{i_1,\ldots,i_k\}|=n-k.
\end{equation}
Since $n<i_{k+1}<\cdots<i_{n}\leq n+s$, we have $n-k\leq s$. We next show, after we recover the blocks indexed by $[n]\backslash\{i_1,\ldots,i_k\}$, the rest blocks can be recovered by peeling decoding without rooting steps. Combining these results together, the total number of rooting steps is at most $s$.

Since we utilize the rooting step to recover blocks indexed by $[n]\backslash\{i_1,\ldots,i_k\}$, we obtain that matrix $\mathbf{M}$ $i$th column  $\mathbf{M}_i=\mathbf{0}$ for $i\in\{1,\ldots,i_1-1\}$. Based on our construction of the $s$-diagonal code, we have $m_{i_1i_1}\neq 0$, which implies $i_1$th block is a ripple. Then we can use the result $\mathbf{\tilde{y}_{i_1}}$ to recover block $\mathbf{y_{i_1}}$ and peel the $i_1$th column, which implies that $\mathbf{M}_{i_1}=\mathbf{0}$. Using the similar process, we can find a ripple $\mathbf{\tilde{y}_{i_2}}$ and peel the $i_2$th column. Continue this process, we can peel the $i_k$th column. 

Here we analyze the complexity of above procedure. During each iteration, the complexity of operation $\mathbf{\tilde{y}}_{j}=\mathbf{\tilde{y}}_{j}-m_{ji}\mathbf{A}_i\mathbf{x}$ is $O(r/n)$. There exists the total $n(s+1)$ the above operations. The complexity from peeling decoding is $r(s+1)$. The complexity in $s$ rooting steps (\ref{eq:borrow}) is $O(rs)$. Therefore, the total complexity is $O(rs)$ and theorem follows.

\end{document}